\documentclass[11pt]{article}
\usepackage{amsfonts,amsmath,amssymb,appendix,pstricks,verbatim,amsthm,enumerate}

\usepackage{setspace,fullpage}

\theoremstyle{plain}
\newtheorem{theorem}{Theorem}
\newtheorem{corollary}{Corollary}
\newtheorem{lemma}{Lemma}

\theoremstyle{definition}

\theoremstyle{remark}

\def\eps{\epsilon}

\title{Truthfulness, Proportional Fairness, and Efficiency\footnote{This work was supported in part by NSF grant CCF0830516.}}

\author{ Richard Cole
         \thanks{{\tt cole@cims.nyu.edu}. Courant Institute, New York University, N.Y.}
        \and
         Vasilis Gkatzelis 
         \thanks{{\tt gkatz@cims.nyu.edu}. Courant Institute, New York University, N.Y.}
        \and
         Gagan Goel
        \thanks{{\tt gagangoel@google.com}. Google Research, New York, N.Y.}
    }

\begin{document}
\maketitle

\begin{abstract}
How does one allocate a collection of resources to a set of {\em strategic}
agents in a fair and efficient manner without using money?
For in many scenarios it is not feasible to use money to compensate agents
for otherwise unsatisfactory outcomes.
This paper studies this question, looking at both fairness and
efficiency measures.

\begin{itemize}
\item
We employ the {\em proportionally fair} solution, which is a
well-known fairness concept for money-free settings.
But although finding a proportionally fair solution is computationally
tractable, it cannot be implemented in a truthful fashion.
Consequently, we seek approximate solutions.
We give several truthful mechanisms which achieve proportional
fairness in an approximate sense. We use a strong notion of
approximation, requiring
the mechanism to give each agent a good approximation of its proportionally fair
utility. In particular, one of our mechanisms provides a better and
better approximation
factor as the minimum demand for every good increases.
A motivating example is provided by the massive privatization auction in the
Czech republic in the early 90s.
\item
With regard to efficiency, prior work has shown a lower bound of 0.5
on the approximation
factor of any swap-dictatorial mechanism approximating a social welfare measure
even for the two agents and multiple goods case.
We surpass this lower bound by designing a non-swap-dictatorial
mechanism for this case.
Interestingly, the new mechanism builds on the notion of proportional fairness.

\end{itemize}
\end{abstract}

\thispagestyle{empty}
\newpage	
\setcounter{page}{1}
\section{Introduction}
How does one allocate a collection of resources to a set of {\em strategic}
agents in a fair and efficient manner without using money? This is a
fundemental problem with many applications since in many scenarios payments
cannot be solicited from agents; for instance, different teams compete for
a set of shared resources in a firm, and the firm cannot solicit payments
from the teams to make allocation decisions. Despite the applicability of
mechanism design without money, much of the work in this area relies on
enforcing payments. Motivated by this, we study the above question in this
paper. We look at both fairness and efficiency measures.

One practical issue that arises in mechanism design without money is to put the utilities of the agents on a common scale.
For example, happiness could mean different things to different people and cannot be
compared as such. When payments can be used, a standard approach is to
measure the utilites in terms of \emph{money}.
In the absence of money, one way to overcome this difficulty is to look for {\em
scale-free} solutions, i.e., if an agent scales her valuations up or down,
the solution should remain the same. When maximizing social welfare (SW), one can achieve this by first normalizing
the values of the agents so that they add up to a common number (1 say) and then maximizing
the welfare with these normalized values. For the case of fairness,
the {\em proportionally fair} solution is a well-known fairness concept that is
scale-free. In brief, a proportionally fair (PF) solution is a Pareto
optimal solution $\cal O$ which compares favorably to any other Pareto
optimal solution ${\cal O}'$ in the sense that when switching from ${\cal O}'$
to $\cal O$, percentage gains outweigh percentage losses.
The PF solution was first proposed in the TCP literature and is used
widely in many practical scenarios~\cite{Kelly-refs}.

Maximizing social welfare with normalized values using truthful
mechanisms was first studied by Guo and Conitzer~\cite{GC10}. 
They studied the mechanism design for the
special case of two players and many items to better understand the
structure of good mechanisms maximizing social welfare. Even in
this special case, they showed that the problem is difficult
by proving that no mechanism from a general class called {\em increasing-price} 
mechanisms that use artificial currency can yield better than $0.5$ approximation;
they left it as an open question to overcome this bound. 
Later on, Han et al.~\cite{HanSTZ11} showed the same negative result of $0.5$ for an even more general 
class of mechanisms called {\em swap-dictatorial}\footnote{This class contains all mechanisms that (randomly) 
choose one of the two bidders and allow her to choose her preferred bundle of items from a predefined set; the other
bidder receives the remaining items. For further discussion on swap-dictatorial mechanisms and their importance in 
money-free mechanism design, see \cite{GC10} and references therein.}
mechanisms. 
In this work, we break this bound of $0.5$. Our main contribution here
is to show a connection between the PF and SW allocations;  then we exploit this connection to give
an interesting non-swap-dictatorial mechanism which overcomes the $0.5$ bound.

We then focus on the design of truthful mechanisms for achieving
proportional fairness. Even for simple instances involving just two
agents and two items, it
is not difficult to show that no truthful mechanism can obtain a PF
solution in general. Hence one has to look for approximate solutions.
In this work, we ask for a strong notion of approximation requiring
the mechanism to give each agent a good approximation of its true
proportionally fair utility.

Also, we note that the PF solution is related to the market
equilibrium prices when agents are assumed to each have a \emph{unit budget}
of some artificial currency. In this case, the PF allocation is the
same as
the allocation at market equilibrium prices, and is captured via
the Eisenberg-Gale program~\cite{EisenbergG59,DevanurPSV02,JV07}. This connection of the PF
solution to market equilibrium prices is of independent interest for
the design of truthful mechanisms in this setting.

A specific real world example of the use of artificial currency for
achieving or approximating fair
outcomes, which motivates our final result, is that of the
privatization auctions
that took place in Czechoslovakia. 
%
In the early 90s, the Czech government sought to privatize the state
owned firms dating
from the then recently ended communist era. The government's goal was
two-fold --- first to distribute shares of these companies to their
citizens in a fair manner and to calculate the market prices of these
companies so that the shares could be traded in the open market after
the initial allocation. To this end, they ran an auction, as described
in~\cite{AH2000}.
Citizens could choose to participate by buying 1000 vouchers at a cost
of 1,000 Czech Crowns, about \$35, a fifth of the average monthly salary.
Over 90\% of those eligible participated.
These vouchers were then used to bid for shares in the available 1,491 firms.
We believe that the PF allocation provides a very appropriate solution
for this example, both to calculate a fair allocation and to compute
market prices.
Our mechanism solves the problem of efficiently discovering approximately PF allocations
in a truthful fashion for such natural scenarios where there is high demand for each resource.

\paragraph{Our results.}
We start in Section~\ref{sec:SWapprox} with the goal of (approximately) maximizing the social welfare
using truthful mechanisms. 
Previous work showed that no non-trivial approximation factors can
be achieved by any truthful mechanism for unbounded numbers of items and bidders~\cite{HanSTZ11}.
This work also showed that no swap-dictatorial mechanism can yield an approximation
factor better than the trivial $0.5$ for the two-bidder case. 
Accordingly, the authors asked whether a
truthful mechanism breaking this bound of $0.5$ exists.
We give a positive answer  
by designing a 
truthful non-swap-dictatorial mechanism closely related to the notion of PF, and proving that it
achieves an approximation factor of at least $0.622$.


In Section~\ref{sec:PFapprox}, we consider the objective of (approximately) maximizing the fraction of her
PF utility that each bidder receives.
Subsection~\ref{sec:2-items} studies the set of instances involving just two items, for which we present
a truthful mechanism which guarantees that all bidders will receive exactly the same
fraction of their PF utility; this fraction is always at least $\frac{n}{n+1}$,
where $n$ is the number of bidders. For small values of $n$, we also describe more elaborate
mechanisms achieving better approximation ratios: $0.828$ instead of $2/3$ for $n=2$ and $0.775$
instead of $3/4$ for $n=3$, both of which appear in Appendix~\ref{improved_twoItem}.
Then, in Subsection~\ref{sec:many-many},
we address the general setting with arbitrary numbers of bidders and items.
We present a truthful mechanism that performs 
increasingly well as the PF prices\footnote{The prices induced by the market equilibrium when all bidders have a unit of budget.} increase.
More specifically, if $p_j^*$ is the PF price of item $j$, then the approximation factor guaranteed
by this mechanism is equal to $\min_j \left(p_j^*/\left\lceil p_j^* \right\rceil\right)$.
It is interesting to note that scenarios such as the privatization auction
mentioned above involve a number of bidders much larger than the number of items;
as a rule, we expect this to lead to high prices and a very good approximation of the participants' PF utilities.

\paragraph{Related Work.}
Our setting is closely related to the large topic of fair division or cake-cutting~\cite{BramsT96,RobertsonWebb98},
which has been studied since the 1940's, using the $[0,1]$ interval as the standard representation
of a cake. Each agent's preferences take the form of a valuation function over this interval and
the valuations of unions of subintervals are additive. The multiple divisible items setting that we focus on is
equivalent to restricting the agents' valuations to piecewise constant functions over the $[0,1]$ interval.
Some of the most common notions of fairness that have been studied in this literature are,
proportionality, envy-freeness, and equitability~\cite{BramsT96,RobertsonWebb98}.\footnote{It is worth distinguishing 
the notion of PF from that of proportionality by noting that the latter is a much weaker notion, directly implied by the former.
Also, note that all the mechanisms that we propose yield envy-free outcomes.}

Despite the extensive work on fair resource allocation, truthfulness considerations have not played a major
role in this literature. Most results related to truthfulness were weakened by the assumption that each agent would be
truthful in reporting its valuations unless this strategy was dominated.
Very recent work~\cite{CLPP10,MosselT10,ZivanDOS10} studies truthful cake cutting variations using the
standard notion of truthfulness according to which an agent may not be truthful unless doing so is a dominant strategy.
Chen et al.~\cite{CLPP10} study truthful cake-cutting with agents having piecewise uniform valuations and
they provide a polynomial-time mechanism that is truthful, proportional, and envy-free. They also design randomized
mechanisms for more general families of valuation functions, while Mossel and Tamuz~\cite{MosselT10} prove the existence
of truthful (in expectation) mechanisms satisfying proportionality in expectation for general valuations.
Zivan et al.~\cite{ZivanDOS10} aim to achieve envy-free Pareto optimal
allocations of multiple divisible goods while reducing, but not eliminating, the agents' incentives to lie.
The extent to which untruthfulness is reduced by their proposed mechanism is only evaluated empirically and
depends critically on their assumption that the resource limitations are soft constraints.
Our work is closely related to the recent papers of Guo and Conitzer~\cite{GC10} and of Han et al.~\cite{HanSTZ11},
who also consider the truthful allocation of multiple divisible goods; their goal is to maximize the social welfare
(We provide more details in Section~\ref{sec:SWapprox}).

Most of the papers mentioned above contribute to our understanding of the trade-offs between either truthfulness
and fairness, or truthfulness and social welfare. Another direction that has been actively pursued is to
understand and quantify the interplay between fairness and social welfare. Caragiannis et al.~\cite{CKKK12}
measured the deterioration of the social welfare caused due to different fairness restrictions, the price of
fairness. More recently, Cohler et al.~\cite{CLPP11} designed algorithms for computing allocations that (approximately)
maximize social welfare while satisfying envy-freeness.


Our results fit into the general agenda of approximate mechanism design without money, explicitly initiated by
Procaccia et al.~\cite{ProcacciaT09}. The use of artificial currency aiming to achieve truthfulness and fairness
in such scenarios, and the notion of \emph{competitive equilibrium with equal incomes} was recently revisited by
Budish~\cite{B10} and Budish et al.~\cite{OSB10}.

Finally, the work of Zhang~\cite{Zhang05}, despite the fact that it does not explicitly refer to
the use of artificial currency, enforces that all the participants' budgets be spent, thus
moving to an equivalent scenario. Unlike our work, they focus on evaluating the equilibrium
points of the resource allocation game rather than on designing truthful mechanisms. 

\section{Preliminaries}\label{sec:prelim}
Let $M$ denote the set of $m$ items and $N$ the set of $n$ bidders.
Each bidder $i\in N$ has a valuation $v_{ij}$ for each item $j\in M$ and each
item is divisible, meaning that it can be cut into arbitrarily small pieces and then allocated
to different bidders. The bidder valuations are scaled so that $\sum_j{v_{ij}}=1$ for each
bidder $i$.

Let $x$ be an allocation of the
resources among the bidders and $v_i(x)$ be the valuation of bidder $i$ for this
allocation. Such an allocation is \emph{Proportionally Fair} (PF) if it is feasible, and for any other
feasible allocation $x'$ the aggregate proportional change to the valuations is not positive, i.e.:
\begin{equation*}
\sum_{i\in N}{\frac{v_i(x')-v_i(x)}{v_i(x)}}\leq 0.
\end{equation*}
We focus on problems where every bidder has additive linear valuations; this means that, if a bidder
is allocated a fraction $f_j$ of each item $j$ that she values at $v_j$, then her valuation for that allocation
will be $\sum_j f_j v_j$.
For such additive linear valuations we can compute a PF allocation in polynomial time by assuming that each bidder
has a \emph{unit budget} of some artificial currency and then using the Eisenberg-Gale convex
program to compute the market equilibrium of the induced market~\cite{DevanurPSV02,JV07}.
This outcome need not be unique, but it provides unique item prices and 
achieved bidder valuations~\cite{EisenbergG59};
we will refer to these induced item prices as PF prices.

Given a valuation bid vector from each bidder (one bid for each item), we want to design mechanisms that output
an allocation of items to bidders. We restrict ourselves to truthful mechanisms, i.e.\ mechanisms
such that any false bid from a bidder will never return her a more valuable allocation.
In designing such mechanisms we consider two objectives:

\begin{itemize}
\item The \emph{SW objective}, which aims to output an allocation $x$ (approximately) maximizing the social welfare, denoted $SW(x) = \sum_i v_i(x)$.
\item The \emph{PF objective}, which aims to output an allocation $x$ (approximately) maximizing the value of
$\rho(x) = \min_{i\in N}\left( \frac{v_i(x)}{v_i(x_{PF})} \right)$, where $x_{PF}$ denotes the PF allocation.
\end{itemize} 

Since maximizing these objectives via truthful mechanisms is infeasible in our setting, we will measure the performance
of our mechanisms based on the extent to which they approximate them. More specifically, when referring to an
approximation factor for the SW objective, this will be the minimum value of the ratio $SW(x)/SW(x^*)$ across all
the relevant problem instances, where $x$ is the output of the mechanism and $x^*$ is the allocation that maximizes SW. For the PF
objective the approximation factor will be the minimum value of $\rho(x)$ across all relevant instances.
\section{Social Welfare Approximation}\label{sec:SWapprox}

The problem of maximizing the SW objective in this model was first approached by Guo and Conitzer~\cite{GC10},
who focused on the interesting case of two-bidder instances, which often draws the attention of researchers
studying efficient or fair allocation. They first showed that no truthful mechanism can
achieve better than a $0.841$ approximation of the optimal social welfare, even for two items.
They then studied a subclass of swap-dictatorial mechanisms and showed that no mechanism from that subclass can
achieve better than a $0.5$ approximation of the optimal social
welfare when $m$, the number of items, is unbounded.
Subsequent work of Han et al.~\cite{HanSTZ11} extended these negative results,
showing that no swap-dictatorial mechanism can achieve better than a $0.5$ approximation for the two-bidder case, and that for the general
setting with $n>2$, no truthful mechanism can achieve better than the trivial $1/m$ approximation.
The main open question that remains in this setting is whether interesting truthful mechanisms for the two-bidder case
exist beyond the class of swap-dictatorial mechanisms and whether such mechanisms can achieve an approximation factor better than $0.5$.

We provide a positive answer to both open questions using the notion of PF as a tool.
We first show that, for two-bidder instances, PF allocations approximate the SW objective very efficiently,
even in the worst case. We then give two truthful mechanisms, one of which is inherently connected
to the PF allocation, and show that combining  the two gives a $0.622$ approximation of the SW objective,
beating the lower bound of $0.5$ that was shown in~\cite{HanSTZ11} for all swap-dictatorial mechanisms.

\subsection{The Social Welfare of Proportionally Fair Allocations}
We start by showing that for two agents $A, B,$ and multiple items, the social welfare of the PF allocation $x_{PF}$ is a
very good approximation of the social welfare achieved by $x^*$, the optimal allocation. Specifically:

\begin{theorem}\label{thm:PFeff}
For two agents and multiple items, $\frac{SW(x_{PF})}{SW(x^*)} \geq \frac{2\sqrt{3}+3}{4\sqrt{3}} \approx 0.933$.
\end{theorem}

\begin{proof}
Assume that the items are ordered in a decreasing fashion w.r.t.\ the $\frac{v_{Aj}}{v_{Bj}}$ ratio. It is easy to verify that
$x^*$ allocates each item $j$ either to agent $A$, or to agent $B$, depending on whether $v_{Aj}$ or $v_{Bj}$ is larger.
We assume that items valued equally by both are allocated to agent $B$, so $x^*$ can be thought of as defining an item $e$
in the ordering mentioned above such that all the items preceding item $e$ are allocated to agent $A$ and all the rest to agent
$B$. This type of ``frontier'' within this ordering, defining each player's allocation, actually occurs with all Pareto
efficient allocations, so the PF allocation also defines such a frontier\footnote{If the PF solution dictates that an item has
to be shared between the two agents, we can just split it into two items such that each one is allocated to a different agent; such a change does not affect the item ordering.}.
W.l.o.g.\ we assume that,  in the items' ordering, the frontier of PF comes before the frontier of the social welfare maximizing allocation.

These two frontiers separate the set of items into three groups.
Slightly abusing notation, let $v_{Ag}$ and $v_{Bg}$ denote the valuations of agents $A$ and $B$ respectively
for group $g\in \{1,2,3\}$, where the group number $g$ is consistent with the initial item ordering.
Note that $v_{A1}/v_{B1} \geq v_{A2}/v_{B2}$,
and  $v_{A2}/v_{B2}\geq v_{A3}/v_{B3}$. The ratio that we are studying can thus be rewritten as follows:

\begin{equation}\label{eq:ratio}
\frac{SW(x_{PF})}{SW(x^*)} ~=~ \frac{v_{A1}+v_{B2}+v_{B3}}{v_{A1}+v_{A2}+v_{B3}} ~=~ 1 - \frac{v_{A2}-v_{B2}}{v_{A1}+v_{A2}+v_{B3}}.
\end{equation}

Let $v_{A2} = kv_{B2}$ for some $k>1$. Then, $v_{A1}\geq kv_{B1}$.
Thus $k(v_{B1}+v_{B2})\leq v_{A1}+v_{A2} \leq 1$, or $k(1-v_{B3})\leq 1$, which yields
that $v_{B3}\geq (k-1)/k$. 
Also, by the definition of PF,
since the PF solution allocates the second group of items to agent $B$, 
$\frac{v_{A2}}{v_{A1}} \leq \frac{v_{B2}}{v_{B2}+v_{B3}}$ which, after substituting for
$v_{A2}$, yields $v_{A1}\geq k(v_{B2}+v_{B3})$.
Since $v_{B2}+v_{B3}=1-v_{B1}$, this inequality can be rewritten as $v_{A1}\geq k (1-v_{B1})$. Adding this inequality to
$v_{A1}\geq kv_{B1}$ yields
$v_{A1}\geq k/2$. Using these lower bounds for $v_{A1}$ and $v_{B3}$ in Equation (\ref{eq:ratio}), we get:

\begin{equation*}
\frac{SW(x_{PF})}{SW(x^*)} ~\geq~ 1 - \frac{(k-1)v_{B2}}{\frac{k}{2}+kv_{B2}+\frac{k-1}{k}} ~=~ 1 - \frac{2k(k-1)v_{B2}}{2k^2v_{B2}+k^2+2k-2} .
\end{equation*}

We can then show that for any value of $v_{B2}$, the right hand side of the inequality is minimized for $k$ as large as possible,
and that $k$ is restricted to be at most $\frac{1}{v_{B2}+0.5}$, implying that
$k= \frac{1+\sqrt{3}}{2}$
in the worst case, thereby proving the theorem (details in Appendix~\ref{app:missing}).

\end{proof}

\subsection{Two Truthful Mechanisms}
\paragraph{Swap-Dictatorial Mechanism.}
Our first mechanism picks one of the two players with 
probability $\frac 12$, and then, based on her bid, provides her
with her most preferred bundle of at most $m/2$ items.
The other player is allocated everything that is left.
Clearly, this mechanism is truthful since a bidder would prefer to bid truthfully if picked,
and her bid does not affect the allocation if she is not picked.
This mechanism achieves a $0.5$ approximation: If $v$ is some bidder's valuation for her allocated bundle if she is picked by the mechanism, 
then the least value that she receives if she is not picked
(and the other bidder is allocated her favorite bundle) is $1-v$. Therefore, both bidders get
an expected value of at least $0.5$, yielding an expected social welfare of at least $1$; the optimal social welfare is of course at most $2$. 
This is tight, as
the following example shows:
there are 4 items
valued by bidder A at $(1 - 2\eps,\eps,\eps/2,\eps/2)$ and by Bidder B at $(\eps,1 - 2\eps,\eps/2,\eps/2)$,
for some small $\eps>0$; the mechanism achieves a $1/2 + \Theta(\eps)$ approximation.
Finally, we note that we can use a deterministic version of this mechanism:
it plays this game twice, once with Bidder A as the dictator and
once with Bidder B in this role.
Each game is played with a set comprising one half of each original item.

\paragraph{Partial Allocation Mechanism.}
We now present an interesting non-swap-dictatorial truthful mechanism which we call the \emph{Partial Allocation} (PA) mechanism.
Let $v_A = v_A(x_{PF})$ and $v_B = v_B(x_{PF})$ denote respectively the fraction of their total valuation that Bidders A and B receive in the PF allocation (remember that the total valuation is equal to 1). 
The mechanism allocates Bidder A a fraction $v_B$ of each of the
items in her PF allocation and, similarly, Bidder B a fraction $v_A$ of each of the items in her PF
allocation.\footnote{Like before, we can split an item into two in order to avoid any sharing of items in the PF allocation.}
Note that, as a direct implication, Bidder A's utility is a $v_B$ approximation of her PF utility
and similarly Bidder B's utility is a $v_A$ approximation. In contrast to the previous mechanism, the types of instances
for which this mechanism performs poorly are the ones where, for example, both bidders value all items equally; this
would cause both bidders to receive only a $0.5$ fraction of their PF allocation, a $0.5$ approximation of the SW objective.
\begin{lemma}
The PA mechanism is truthful.
\end{lemma}
\begin{proof}
We show that Bidder A is truthful; the same argument applies to Bidder B.
We rescale the valuations of Bidder A so that her new valuation for the PF allocation is
$\bar{v}_A = v_B$.
Suppose Bidder A changes her bid and, in the changed PF allocation, gains allocation of value $g_A$
and loses value $l_A$, according to the scaled valuations. Let the gains  and losses to Bidder B
have values $g_B$ and $l_B$.
Since $\bar{v}_A = v_B$ and given the definition of PF, we get $l_B \ge g_A$ and $l_A \ge g_B$, for otherwise either one
of these shifts of allocation between A and B would have lead to a higher percentage increase than decrease.
The rescaled valuation of Bidder A after changing her bid would be:

\begin{eqnarray*}
(v_B - l_B + g_B)(\bar{v}_A - l_A + g_A)  & = &
 v_B \bar{v}_A + v_B (g_A - l_A) + \bar{v}_A (g_B - l_B) + (g_B - l_B)(g_A - l_A)  \\
& = &  v_B[\bar{v}_A - (l_A - g_A) - (l_B - g_B)] - (g_A - l_A)(l_B - g_B)  \\
& = &  v_B[\bar{v}_A - (l_B - g_A) - (l_A - g_B)] - (g_A - l_A)(l_B - g_B)  \\
& \le &  v_B\bar{v}_A - (g_A - l_A)(l_B - g_B).
\end{eqnarray*}

Now, if $g_A \ge l_A$, then as $l_B \ge g_A$,
and
$l_A \ge g_B$,
it follows that
$l_B \ge g_B$;
thus
the new valuation is
at most
the original $v_B\cdot \bar{v}_A$.
A similar argument applies if $g_B \ge l_B$.
Otherwise, $g_A \le l_A$ and $g_B \le l_B$ and the reduction in value for Bidder A as
is at least
$v_B[(l_A - g_A) + (l_B - g_B)] + (g_A - l_A)(l_B - g_B)
= (l_A - g_A)[v_B - (l_B - g_B)] + (l_B - g_B) \ge 0$,
as $l_B \le v_B$.
\end{proof}

\subsection{The Hybrid Mechanism}
We now provide a non-swap-dictatorial truthful mechanism that outperforms all swap-dictatorial mechanisms, giving a 0.622 approximation of the SW objective.
Our mechanism is a combination of the swap-dictatorial mechanism and the PA mechanism described in the previous subsection.
A randomized version of it would run the swap-dictatorial mechanism with probability $0.5$ and the PA mechanism with the remaining $0.5$ probability.
Again, this can be implemented in a deterministic fashion by spliting all items in half and using a different
mechanism for each set of halves. The benefit of combining them comes from the fact that one performs well when the other one does not.

\begin{lemma}\label{lem:Hybrid}
The hybrid mechanism outputs an allocation $x_m$ always satisfying $\frac{SW(x_m)}{SW(x_{PF})}\geq \frac{2}{3}$.
\end{lemma}

\begin{proof}
(Sketch.)~
The swap-dictatorial mechanism always provides a social welfare of at least $1$.
The PA mechanism provides a social welfare of $2v_Av_B$. Thus, $\frac{SW(x_m)}{SW(x_{PF})} = \frac{0.5 + v_Av_B}{v_A+v_B}$.
We show that for all $(v_A,v_B)\in [0.5 , 1]^2$ this is minimized for $v_A=0.5$ and $v_B=1$ (details in Appendix~\ref{app:missing}).
\end{proof}

This lemma combined with Theorem \ref{thm:PFeff} immediately implies:

\begin{theorem}
The hybrid mechanism always satisfies $\frac{SW(x_m)}{SW(x^*)}\geq 0.933 \cdot \frac{2}{3} \approx 0.622$.
\end{theorem}

\section{Proportional Fairness Approximation}\label{sec:PFapprox}

Despite the fact that maximizing social welfare is a very natural objective, one can quickly
verify that allocations that are very efficient in terms of social welfare can be very unfair
to some bidders. Dealing with problems such as the one that arose with the Czech privatization 
auctions~\cite{AH2000} calls for solutions that are fair. In what follows, we use the PF allocation
as a benchmark for the ``fair share'' that each bidder should be receiving and we aim to provide
every bidder with a good approximation of his value for that share.
We start by focusing on the case of two items which helps build some intuition for our solution
to the general case that follows.

\subsection{Two Items}\label{sec:2-items}

\paragraph{Proportionally fair allocation for two items.}
For instances involving just two items $t$ and $b$ (or top and bottom), for simplicity we choose
one of the two items (w.l.o.g.\ $b$) and rescale the valuations of each bidder so that they all have a
valuation of 1 for item $b$. (If a bidder has zero valuation for item $b$, then her valuation
for item $t$ is set to be $\infty$.) We then sort the bidders in decreasing order of their valuation
for item $t$ (breaking ties arbitrarily). The proof of the following lemma (in Appendix \ref{app:missing})
shows that for two-item instances there always exists a PF allocation such that at most one bidder is allocated
parts of both items. The proof also shows that if such a bidder exists and her valuation for item $t$
is $v$, then her valuation for the PF allocation will be $\frac{v+1}{n}$. Such a bidder will also
be defining the relative value of the PF prices of the two items (the PF price of $t$ is $v$ times that of $b$),
so we will henceforth call her the \emph{Ratio Defining Bidder}, denoted $R_b$.

\begin{lemma}\label{lem:TwoItems}
For two item instances, there always exists a PF allocation with at most one $R_b$.
\end{lemma}

\paragraph{Mechanism for many bidders and two items.}

In this scheme \emph{every} bidder receives a fraction of just one item.
The mechanism can be thought of as a tie breaking rule that doesn't allow the existence of a bidder
being allocated parts of more than one item. More specifically, given the bidders' bids,
the mechanism computes a PF allocation with at most one $R_b$\footnote{This can be done by using binary search
over the bidder ordering until $R_b$, if one exists, is found.} and if one exists for that allocation,
we instead force her to either equally share item $t$ with the {\sc Top} bidders or item $b$
with the {\sc Bottom} bidders (respectively the bidders before and after $R_b$ in the above ordering).
Among these two options, the mechanism chooses the one that maximizes $R_b$'s utility.
Let $\rho$ denote the fraction of her PF utility that $R_b$ receives.
Every other bidder's allocation is reduced
as needed
so as to provide that same $\rho$ fraction.
We call this the \emph{Single Item} (SI) mechanism.

\begin{lemma}
\label{lem:SIMapprox}
With $n$ bidders, the SI mechanism achieves an approximation factor of $\frac{n}{n+1}$.
\end{lemma}
\begin{proof}
The worst $\rho$ occurs when $R_b$ achieves the same utility with either option. If $R_b$ is the $k$-th
bidder in the ordering, this means that she values being allocated $1/k$ of the top item equally with
$\frac{1}{n-k+1}$ of the bottom item, i.e.\ $\frac vk = \frac{1}{n-k+1}$, or $ v = \frac{k}{n-k+1}$.
The approximation factor then,
on substituting for
$v$,
is given by:
\begin{equation*}
\frac{v/k}{(v+1)/n} = \frac{1/(n-k+1)}{[(n+1)/(n-k+1)]/n} =\frac{n}{n+1}.
\end{equation*}
\end{proof}

\begin{lemma}
\label{lem:SIMtruth}
The SI mechanism is truthful. (Proof in Appendix \ref{app:missing}.)
\end{lemma}

This mechanism performs really well as the number of bidders increases, but can be
further improved for small values of $n$.
In Appendix~\ref{improved_twoItem} we present an interesting and more efficient mechanism for the special case of $n=2$,
which achieves an approximation factor of $0.828$ instead of $2/3$, and a similar mechanism for $n=3$, achieving an 
approximation factor of $0.775$ instead of $3/4$.
As $n$ increases beyond $n=3$, such more elaborate schemes provide, at best, very modest gains compared to the
$\frac{n}{n+1}$ approximation mechanism presented above.

\newcommand{\p}{\mbox{\boldmath $p$}}
\newcommand{\q}{\mbox{\boldmath $q$}}

\subsection{Many Bidders and Many Items}\label{sec:many-many}

Using the intuition acquired from the two-item case, we now describe the {\em Strong Demand Matching} mechanism (SDM) for the general case.
Informally speaking, SDM provides every bidder with a unit budget and then aims to discover item prices such that 
the demand of each bidder can be satisfied using (a fraction of) just one item. This approach resembles the SI mechanism, but the structure
of the PF solution can now be much more complicated. This mechanism tackles the problem of allocating a collection of goods for the very
natural set of problem instances for which the bidders may have arbitrary valuation functions, yet no item is undemanded if its price is low.
In what follows we describe the process by which SDM increases the prices of
overdemanded items in a fashion that maintains truthfulness and yields prices that are very close to the PF prices.


Let $p_j$ denote the price of item $e_j$, and let the \emph{bang per buck} that Bidder $i$ gets 
from item $e_j$ equal $v_{ij}/p_j$. We say that item $e_j$ is an MBB item of Bidder $i$ if Bidder $i$ gets the 
maximum bang per buck from that item\footnote{Note that for each bidder there could be multiple MBB items.}.
For a given price vector $p$, let the demand graph $D(p)$ be a bipartite graph with bidders on one side and items on the other, such that
an edge between Bidder $i$ and item $e_j$ exists if and only if $e_j$ is an MBB item of Bidder $i$.
We call $c_j=\lfloor p_j\rfloor$ the \emph{capacity} of item $e_j$ when its price is $p_j$, and we say an assignment of bidders to items is \emph{valid} 
if it matches each bidder to at most one item and no item $e_j$ is matched to more than $c_j$ bidders.
Given a valid assignment $A$, we say an item $e_j$ is \emph{reachable} from Bidder $i$ if there exists an alternating 
path $(i, j_1, i_1, j_2, i_2,\cdots, j_k, i_k, j)$ in the graph $D(p)$ such that edges $(i_1, j_1),\cdots ,(i_k, j_k)$ 
lie in the assignment $A$. Finally, let $d(R)$ be the collection of bidders with all their MBB items in set $R$. \\

\noindent
The SDM mechanism initializes all item prices to $p_j=1$ and iterates as follows:

\begin{enumerate}
\item
Find a valid assignment that maximizes the number of matched bidders.\\ 
If all the bidders are matched, conclude with Step 3.

 \item
Let $U$ be the set of bidders who are not matched in Step 1. \\
Let $R$ be the set of all items reachable from bidders in the set $U$.\\
Raise the price of each item $e_j$ in $R$ from $p_j$ to $x\cdot p_j$,\\
where $x\geq 1$ is the minimum value for which one of the following events takes place:

\begin{enumerate} 
\item
The price of an item in $R$ reaches an integral value. If this happens, repeat Step 1. 

\item
For some bidder $b_i \in d(R)$, her set of MBB items increases, causing $R$ to grow:

\begin{enumerate} 
\item
If for each item $e_j$ added to $R$, the number of bidders matched to it equals $c_j$,\\ continue with Step 2.
\item
If some item $e_j$ added to $R$ has $c_j$ greater than the number of bidders matched to it, continue with Step 1.
\end{enumerate}
\end{enumerate}

 \item
Every bidder matched to some item $e_j$ is allocated a fraction $1/p_j$ of that item.

 \end{enumerate}


\noindent
It remains to explain how to carry out Step 2. Set $R$ can be found using a breadth-first-search like algorithm.
To determine when (a) is reached, we just need to know the smallest $\lceil p_j\rceil/p_j$
ratio over all items whose price is being increased.
For (b), we need to calculate, for each bidder in $d(R)$, the ratio of the \emph{bang per buck}
for her MBB items and for the items outside the set $R$.

\paragraph{Running time.}
If $c(R)=\sum_{j\in R} c_j$ denotes the total capacity in $R$, it is not difficult to see that if $U$ is non-empty, $|d(R)| > c(R)$.
Note that each time either event (a) or event (b)-ii occurs, $c(R)$ increases by at least 1, 
and thus, using the alternating path from a bidder in the set $U$ 
to the corresponding item, we can increase the number of matched bidders by at least 1;
this means that this can occur at most
$n$ times. The only other events are the unions resulting from (b)-i.
There can be at most $\min(n,m)$ of these, and they are followed by either Step (a) or (b)-ii.
Thus there are $O(n*\min(n,m))$ iterations of Step (b)-i and $O(n)$ iterations of Steps 1 and (b)-ii.

\paragraph{Correctness.}
Let $p^*$ be the PF prices and let $q$ be the prices computed by the
algorithm.

\begin{lemma}
\label{lem:price-rounding}
 Let $f= \max_j \lceil p^*_j\rceil/p^*_j$. Then $q \leq f p^*$.
\end{lemma}

\begin{proof}

 First note that at prices $fp^*$, the MBB items for each bidder
are the same as at prices $p^*$. It is not difficult to see that at prices $fp^*$ every bidder can
be allocated to exactly one item from among her
MBB items such that the number of bidders allocated to an item $e_j$ is less than or equal to $fp^*_j$. To show this, consider a PF allocation.
Form the following graph on items and bidders --- add an edge between a bidder and an item if a portion of this item is assigned to this bidder in the PF solution.
If there exists a cycle in this graph, one can remove an edge in this cycle by reallocating along the cycle while maintaining
the utility of every bidder. Hence there is a PF allocation in which this graph is a forest. Now for a given tree, root it at an arbitrary bidder. 
For each Bidder $b$ in this tree, assign it to one of its child items,
if any, and otherwise to its parent. 
The result is that for each item $e_j$, at most $\lceil{p_j}\rceil$ bidders will be assigned to it.

Now, suppose that some $q_j > f p^*_j$. Consider the first time $t$ at which some price $q_i$ starts to increase
from $fp^*_i$.
Let $S$ be the set of items $e_i$ whose price is currently  $fp^*_i$. We will show that no item in set $S$ will be
 part of the set $R$, and hence the prices of these items will not increase. Let $T_q$ and $T_{fp^*}$ be the sets 
 of bidders who have edges to some item in the set $S$ at the current prices and at prices $fp^*$, respectively. 
 Clearly $T_{q} \subseteq T_{fp^*}$. Also suppose a bidder $b \in T_{fp^*}$ has an edge to an item outside set $S$ 
 at prices $fp^*$; this means that $b \not \in T_{q}$ at the current prices as $b$ will strictly prefer the item 
 outside the set $S$. Thus if $b \in T_{q}$, this implies that $b$ has no edges outside the set $S$ at prices $fp^*$, 
 and so $b$ was allocated to some item in set $S$ at prices $fp^*$. Since we know that at prices  $fp^*$ all the 
 bidders can be allocated to some item, this implies that $|T_{q}| \leq c(S)$. Thus even at current prices, all the 
 bidders in $T_{q}$ can be allocated to items in set $S$, and hence no item in set $S$ can be part of the set $R$.
\end{proof}

\paragraph{Truthfulness.}
We argue by contradiction.
Suppose that some Bidder $b$ were not truthful in the above algorithm, which
we name algorithm $\cal A$.
First, we consider an alternate algorithm $\cal A'$, and show that it is
a dominant strategy for $b$ to be
truthful in $\cal A'$.
Then we show that $\cal A$ and $\cal A'$ produce the same outcomes, and
consequently $b$ should also be truthful in $\cal A$.
$\cal A'$ proceeds as follows.
It begins by running algorithm $\cal A$ but with $b$ absent
(the first run), yielding prices $p'$.
Then it runs $\cal A$ on all $n$ bidders, but starting from prices $p'$
(the second run).

\begin{lemma}
\label{lem:var-alg-thruth}
$b$ is truthful in algorithm $\cal A'$.
\end{lemma}
\begin{proof}

Suppose the second run of the algorithm ends when $b$ can be matched using an alternating path that ends at an item $e$. 
Suppose at this point, the price of an item $e_j\in R$ is $p_j$. It is easy to see that bidder $b$ has no incentives to 
lie to obtain an item that is not in $R$ as the prices of these items is completely defined by other bidders. Suppose 
that by lying, Bidder $b$ is able to get an item $e_j\in R$ at a price $p'_j < p_j$. Suppose that this happens with an 
alternating path that starts at $b$ and ends at some item $e'$. Now if this path existed in the truthful scenario when 
the price of item $e_j$ reaches $p'_j$, then even in the truthful scenario $b$ would have been matched when the price of 
item $e_j$ is $p'_j$. Thus this path doesn't exist in the truthful scenario when the price of item $e_j$ is $p'_j$. But why 
does this path not exist in the truthful scenario? It must be the case that some item on this path has a higher price 
than the price in the lying scenario. A higher price on this item means that in the truthful scenario item $e_j$ would 
have been matched via some alternating path ending at $e'$ before the price of item $e_j$ reached $p'_j$, a contradiction.
\end{proof}

\begin{lemma}
\label{lem:equl-outcomes}
$\cal A$ and $\cal A'$ have the same outcome.
\end{lemma}
The proof of this Lemma is similar to that of Lemma~\ref{lem:price-rounding} and is deferred to  Appendix~\ref{app:missing}.
\begin{corollary}
\label{cor:A-truthful}
$\cal A$ is truthful.
\end{corollary}

\begin{theorem}
The SDM mechanism achieves an approximation factor of $\rho = \min_j\left(p^*_j/\lceil{p^*_j}\rceil\right)$.
If $\min_j p^*_j = k$, this is an approximation factor of at least $\frac{k}{k+1}$.
\end{theorem}
\begin{proof}
If a bidder is allocated a portion of item $e_j$, she receives a $1/q_j$ fraction.
But by Lemma~\ref{lem:price-rounding},
$ 1/q_j \ge 1/(fp^*_j) \ge 1/\lceil p^*_j \rceil$.
In value, her PF allocation equals a $1/p^*_j$ fraction of item $e_j$.
Thus she achieves an approximation factor of $p^*_j/\lceil p^*_j \rceil$.
The result follows on minimizing over all bidders.
\end{proof}

\bibliographystyle{amsplain}
\bibliography{MARA}



\appendix
\section{Omitted Proofs}\label{app:missing}

\begin{proof}[of Theorem \ref{thm:PFeff}]
Assume that the items are ordered in a decreasing fashion w.r.t.\ the $\frac{v_{Aj}}{v_{Bj}}$ ratio. It is easy to verify that
$x^*$ allocates each item $j$ either to agent $A$, or to agent $B$, depending on whether $v_{Aj}$ or $v_{Bj}$ is larger.
We assume that items valued equally by both are allocated to agent $B$, so $x^*$ can be thought of as defining an item
in the ordering mentioned above such that all the items preceding this item are allocated to agent $A$ and all the rest to agent
$B$. This type of ``frontier'' within this ordering, defining each player's allocation, actually occurs with all Pareto 
efficient allocations, so the PF allocation also defines such a frontier\footnote{If the PF solution dictates that an item has
to be shared between the two agents, we can just split it into two items such that each one is allocated to a different agent without affecting the ordering.}.
W.l.o.g.\ we assume that the frontier of PF comes before the frontier of the social welfare maximizing allocation in the items' ordering.

These two frontiers separate the set of items into three groups.
Slightly abusing notation, let $v_{Ag}$ and $v_{Bg}$ denote the valuations of agents $A$ and $B$ respectively 
for group $g\in \{1,2,3\}$, where the group number $g$ is consistent with the initial item ordering.
Note that $v_{A1}/v_{B1} \geq v_{A2}/v_{B2}$,
and  $v_{A2}/v_{B2}\geq v_{A3}/v_{B3}$. The ratio that we are studying can thus be rewritten as follows:

\begin{equation}\label{eq:ratioApp}
\frac{SW(x_{PF})}{SW(x^*)} ~=~ \frac{v_{A1}+v_{B2}+v_{B3}}{v_{A1}+v_{A2}+v_{B3}} ~=~ 1 - \frac{v_{A2}-v_{B2}}{v_{A1}+v_{A2}+v_{B3}}.
\end{equation}

Let $v_{A2} = kv_{B2}$ for some $k>1$. Then, $v_{A1}\geq kv_{B1}$.
Also, by the definition of PF,
since the PF solution allocates the second group of items to agent $B$, then
$\frac{v_{A2}}{v_{A1}} \leq \frac{v_{B2}}{v_{B2}+v_{B3}}$ which, after replacing for 
$v_{A2}$, gives $v_{A1}\geq k(v_{B2}+v_{B3})$.
Since $v_{B2}+v_{B3}=1-v_{B1}$, this inequality can be rewritten as $v_{A1}\geq k (1-v_{B1})$. Adding this inequality to
$v_{A1}\geq kv_{B1}$ yields
$v_{A1}\geq k/2$. Also, according to the definition of $k$, we get that $k(v_{B1}+v_{B2})\leq v_{A1}+v_{A2} \leq 1$, or $k(1-v_{B3})\leq 1$, which yields
that $v_{B3}\geq (k-1)/k$. Using these lower bounds for $v_{A1}$ and $v_{B3}$ in Equation (\ref{eq:ratioApp}), we get:

\begin{equation}\label{ineq:ratioApp}
\frac{SW(x_{PF})}{SW(x^*)} ~\geq~ 1 - \frac{(k-1)v_{B2}}{\frac{k}{2}+kv_{B2}+\frac{k-1}{k}} ~=~ 1 - \frac{2k(k-1)v_{B2}}{2k^2v_{B2}+k^2+2k-2} .
\end{equation}

The lower bound implied by this inequality is minimized when the fraction on the right hand side is maximized. Assuming that $v_{B2}$ is 
fixed,
we take the partial derivative w.r.t.\ $k$, which is equal to:
\begin{equation*}
 \left( \frac{2k(k-1)v_{B2}}{2k^2v_{B2}+k^2+2k-2} \right)_k' ~=~ \frac{\left( 2k^2 v_{B2} + 3k^2 -4k +2 \right)2v_{B2}}{(2k^2v_{B2}+k^2+2k-2)^2}.
\end{equation*}
It is easy to verify that this is positive because $3k^2-4k +2>0$ for any value of $k$. This means that
for any value of $v_{B2}$, the fraction is maximized 
when $k$ is as large
as possible. But we know that $kv_{B2} = v_{A2} \leq 1-v_{A1}$,
and since $v_{A1}\geq k/2$, this yields $k \leq \frac{1}{v_{B2}+0.5}$. 
Thus to maximize the fraction we
let $k = \frac{1}{v_{B2}+0.5}$, or
$v_{B2} = \frac{2-k}{2k}$. 
Substituting for
$v_{B2}$
in
Inequality (\ref{ineq:ratioApp}) 
yields:
\begin{equation*}
\frac{SW(x_{PF})}{SW(x^*)} ~\geq~ 1- \frac{-k^2+3k-2}{4k-2},
\end{equation*}
with the right hand side minimized for 
$k=\frac{1+\sqrt{3}}{2}$,
which proves the theorem.
\end{proof}

\begin{proof}[of Lemma \ref{lem:Hybrid}]
Note that the dictatorial mechanism always provides a social welfare of at least $1$ if used (each agent gets a value of exactly $0.5$).
The second mechanism provides a social welfare of $2v_Av_B$. We can therefore express the ratio of the social welfare  of an allocation $x_m$ that
is the outcome of combining these two mechanisms, over the social welfare of the PF allocation as $\frac{SW(x_m)}{SW(x_{PF})} = \frac{0.5 + v_Av_B}{v_A+v_B}$.

We only need to show that $\frac{0.5 + v_Av_B}{v_A+v_B}\geq \frac{2}{3}$ for all $v_A,v_B$ combinations in $[0.5 , 1]^2$
(since it must be the case that $v_A, v_B \ge 0.5$ in the PF allocation). We treat $v_B$
as a given constant and take the derivative w.r.t.\ $v_A$. This gives a derivative of $\frac{v_B^2-0.5}{(v_A+v_B)^2}$, which is negative for
$v_B^2 < 0.5$, positive for $v_B^2 > 0.5$, and zero for $v_B^2 = 0.5$.
This means that, given a value of $v_B<\sqrt{0.5}\approx 0.7$, the ratio is minimized for $v_A=1$,
and given a value of $v_B>\sqrt{0.5}\approx 0.7$, the ratio is minimized for $v_A=0.5$. Given the symmetry of the ratio w.r.t.\ to $v_A$ and $v_B$, the exact
same argument applies for the values of $v_B$ minimizing the ratio given a value of $v_A$. We 
conclude that the ratio is 
minimized
either when
$v_A=0.5$ and $v_B=1$
or when $v_A=1$ and $v_B=0.5$.
In either case, the ratio evaluates to $\frac 23$, proving the theorem.
\end{proof}

\begin{proof}[of Lemma \ref{lem:TwoItems}]

Let the PF prices be $p_t$ and $p_b$, and note that, since $p_t +p_b = n$,
either both prices are
integers or neither is. If both prices are integers, it is easy to see that the following allocation
is PF:
the first $p_t$ bidders are assigned a $1/p_t$ fraction of item $t$ and the remaining $p_b$ bidders are assigned a $1/p_b$
fraction of item $b$. If neither of the prices is an integer, then we get a PF allocation by giving each of the first
$\lfloor p_t \rfloor$ bidders (the {\sc Top} bidders) a $1/p_t$ fraction of item $t$, Bidder $\lceil p_t \rceil$  a
$\frac{p_t -\lfloor p_t \rfloor}{p_t}$
fraction of item $t$ and a $\frac{p_b -\lfloor p_b \rfloor}{p_b}$ fraction of item $b$, and each of the remaining bidders (the {\sc Bottom} bidders) a
fraction of $1/p_b$ of item $b$.

Then, if a bidder being
allocated a fraction of both the top and the bottom items in the PF solution exists and
her valuation for the top item is $v$, notice that in that PF solution
the allocation of every  {\sc Top} bidder comprises a fraction
of the top item alone; similarly, the allocations for {\sc Bottom} bidders
come from the bottom item alone.
If there are $k-1$  {\sc Top} bidders, then
the price of item $t$ in the PF solution would be equal to $k-1+x$, and that of
item $b$ would be equal to $n-k+1-x$, where $x<1$ is the amount Bidder $k$ spends on item $t$.
Notice that since Bidder $k$ is interested in both items at these prices, this means
that her valuations are the ones defining the ratio between the two prices (henceforth, we call her the
\emph{Ratio Defining Bidder}, denoted $R_b$), i.e.:
\begin{equation}\label{valueRange}
 \frac{k-1+x}{n-k+1-x}=v
\end{equation}

Thus:
\begin{equation}
 x=\frac{(n-k+1)v-(k-1)}{v+1}.
\end{equation}
Every {\sc Top} bidder gets a fraction
of item $t$ equal to the ratio of her budget over
the final price of that item, i.e.\ a fraction equal to $\frac{v+1}{vn}$.
Similarly, every {\sc Bottom} bidder
gets a fraction of the bottom item equal to $\frac{v+1}{n}$.
Finally, $R_b$ gets a fraction of item $t$
equal to $\frac{n-k+1}{n}-\frac{(k-1)}{vn}$ and a fraction of item $b$ equal to $\frac{k}{n}-\frac{(n-k)v}{n}$.
Notice that if $R_b$ were truthful,
then this allocation would offer her utility equal to $\frac{v+1}{n}$, equal to that of every {\sc Bottom}
bidder and of every {\sc Top} bidder.
This utility is also equal to the utility she would obtain from a $\frac{v+1}{vn}$ fraction of the top item, or a
$\frac{v+1}{n}$ fraction of the bottom item.

\end{proof}

\begin{proof}[of Lemma~\ref{lem:SIMtruth}]
To verify that this mechanism is truthful, we consider the strategies of any {\sc Top} bidder,
any {\sc Bottom} bidder and those of a possible ratio defining bidder.

\smallskip
\noindent
{\sc Top} bidders (the same arguments apply for {\sc Bottom} bidders as well):
 \begin{itemize}
\item
A {\sc Top} bidder overbids: then her allocation is unchanged.
\item
A {\sc Top} bidder underbids but continues to overbid $R_b$:
then her allocation is unchanged.
\item
A {\sc Top} bidder underbids $R_b$ and $R_b$ had been forced to
share the \emph{top} item: then the only possible change is that the {\sc Top} bidder
ends up sharing the bottom item which is something not
even  $R_b$ wanted.
\item
A {\sc Top} bidder underbids $R_b$ and $R_b$ had been forced to
share the \emph{bottom} item. If $R_b$ does not move then the $\rho$ for the
{\sc Bottom} bidders is reduced, and hence it is reduced for everyone;
while if $R_b$ moves, then the best option for the {\sc Top} Bidder is to
become the ratio defining bidder, and the only choice with the underbid is for the
{\sc Top} Bidder to be sharing the bottom item.
But her $\rho$ ($=n/[(v+1)(n-k+1)]$) is smaller than the approximation of $R_b$ which is
the approximation she achieved before.
\end{itemize}

\smallskip
\noindent
Ratio Defining Bidder (w.l.o.g.\ assume that she had been forced to share the top item):
\begin{itemize}
\item
There is no point in overbidding.
\item
By underbidding, one possible outcome is that this bidder shares the bottom item along with those who had already
been sharing it, which of course by definition is worse for her.
\item
By underbidding one of the {\sc Bottom} bidders this bidder ceases to be ratio defining. The interesting case is when
there is a new ratio defining bidder and she is forced to go up. In that case, the $\rho$ for that bidder ($vn/k(v+1)$)
is smaller than the approximation of the original ratio defining bidder and everyone suffers this same $\rho$.

\end{itemize}
\end{proof}

\begin{proof}[of Lemma~\ref{lem:equl-outcomes}]
Suppose not for a contradiction. W.l.o.g. suppose that some price is higher in the outcome of $\cal A'$ (if not, switch the
roles of $\cal A$ and $\cal A'$).

In the run of the mechanism $\cal A'$, consider a time when no item has exceeded the final price given by the mechanism $\cal A$ but some items have reached that price. Let $S$ be the set of items at this point that have prices equal to their final prices in $\cal A$.  We will show that no item in set $S$ will be part of the set $R$ from this time onwards in the mechanism $\cal A'$, and hence the prices of these items will not increase in $\cal A'$. Let $T'$ and $T$ be the set of bidders who have edges to some item in the set $S$ at the current prices in $\cal A'$, and at the final prices in $\cal A$, respectively. Clearly, $T'\subseteq T$. Also suppose a bidder $b \in T$ has an edge to an item outside set $S$ at the final prices of $\cal A$; this means that $b \not \in T'$ at the current prices as $b$ will strictly prefer the item outside the set $S$. Thus if $b \in T'$, this implies that $b$ has no edges to an item outside the set $S$ at the final prices given by $\cal A$, and so  in $\cal A$, $b$ was allocated to some item in set $S$. Since we know that at the final prices given by $\cal A$, all the bidders can be allocated to some item, this implies that $|T'| \leq c(S)$. Thus, even at current prices, all the bidders in $T'$ can be allocated to items in the set $S$, and hence no item in set $S$ can be part of the set $R$. Thus no item can have higher final price in $\cal A'$ than in $\cal A$. A contradiction.
\end{proof}

\section{Improved Mechanisms for Two Items}\label{improved_twoItem}

\subsection{Two Bidders}

We start by assuming that both bidders are interested in both items,
i.e.\ have non-zero valuation for both of them;
we touch on the other cases at the end of this section.

Just
as
we did for the many bidder case, we scale their valuations so
that
both
bidders
assign
a value of 1 to item $b$.
Then, let $u$ and $v$ be the
respective
valuations of Bidder A and Bidder B
for item $t$. W.l.o.g.\ we can assume that $u>1$ and $u>v$, which means that
Bidder $A$ will never be $R_b$ and Bidder $B$ will be $R_b$ only if $v>1$.
If no $R_b$ exists, then
the
PF allocation is clearly truthful, but when
Bidder $B$ is
the
Ratio Defining Bidder, our 2-bidder 2-item mechanism defines the
final allocation as a function of just $v$. More specifically, Bidder B gets a
fraction $b(v)=\frac{1}{v}$ of item $b$ and a fraction $t(v)=\frac{1}{2}-\frac{1}{2v^2}$
of item $t$. Finally, Bidder A is allocated all of item $t$ that is not allocated to
Bidder B. The intuition behind this mechanism is that, if Bidder B overbids regarding
her
valuation for item $t$, then she loses part of item $b$.

\begin{theorem}
\label{thm:2-item-2-bidder}
The 2-bidder, 2-item mechanism is truthful and achieves an approximation
factor of $2\cdot(\sqrt 2 - 1)\approx 0.828427$ of the PF objective.
\end{theorem}

\begin{proof}

Notice that the valuation of Bidder B
would equal $t(\bar{v})v+b(\bar{v})$ if her bid is $\bar{v}$,
and for the mechanism to be truthful we need $t(\bar{v})v+ b(\bar{v})\leq t(v)v+b(v)$ for all $\bar{v}$. It is easy to verify
that this 
holds
for our mechanism since the partial derivative of the left hand side w.r.t.\ $\bar{v}$ is equal to $\frac{v-\bar{v}}{\bar{v}^3}$,
and hence $v$ is the optimal solution for Bidder B.

Regarding
the approximation ratio, the  
PF allocation gives a fraction of $\frac{1}{2}-\frac{1}{2v}$ of item $t$ and all of item $b$
to Bidder B, 
yielding
an approximation factor of $\frac{v^2+1}{v^2+v}$,
which is minimized for $v=1+\sqrt{2}$, giving $\rho= 2 \cdot (\sqrt 2 - 1) \approx 0.828427$.
\end{proof}

\paragraph{Other cases.}
If just one bidder is interested in both items, we can view it as the previous case with $u= \infty$
for the bidder interested in only one item.
If the bidders are each interested in only one item, if these are distinct items, the bidders are each
allocated the item they want in full. If it is the same item, they each receive half of it.
In all these cases, $\rho \ge 2\cdot(\sqrt 2 - 1)$.

\subsection{Three Bidders}

\begin{theorem}
\label{thm:3bid-2item}
There is a truthful mechanism for 3 bidders and 2 items achieving an approximation factor of
$(12 -\sqrt{12})/11 \approx 0.77599$.
\end{theorem}
\begin{proof}
Our goal is to find two functions $t(v)$ and $b(v)$ defining the fractions of the top and bottom items respectively that $R_b$
will be assigned as a function of her valuation $v$ for the top item.
Before we move on to define such functions, we note that if they are functions of $v$ alone,
and if $R_b$ obtains a fraction $\rho$ of her utility in the
PF solution, then every other
bidder will have to obtain exactly the same fraction $\rho$ of her PF utility.
Were this not the case, we could easily construct an example where a bidder that gets a different approximation
has valuation $v-\epsilon$ or $v+\epsilon$ (i.e.\ she is practically the same bidder as $R_b$).
In that case, the bidder
with the worse approximation would have the option of changing her bid by at most $2\epsilon$
and securing the better approximation factor.
Therefore every bidder
must be offered the same fraction $\rho$ of her PF allocation utility.

Now we consider the case when the middle bidder is $R_b$.
We assume w.l.o.g.\ that this bidder prefers the top item ($v\geq1$).
Then, by Equation~(\ref{valueRange}), $v\in [1,2)$.
Following the approach used for the case of two bidders, we consider the following families of functions:
\[t(v)=\alpha - \beta \frac{1}{v^2}~~~ \text{ and } ~~~b(v)=\gamma \frac{1}{v}-\delta,\]
where $\alpha, \beta, \gamma$ and $\delta$ are non-negative constants
which we will choose so as to maximize $\rho$ while ensuring both truthfulness and
that the resulting solution can be allocated.
More specifically, for $R_b$ to be truthful, the utility function $t(\bar{v})v+b(\bar{v})$ of $R_b$ has to be maximized at $\bar{v}=v$,
i.e.\ it is a necessary condition that:
\begin{equation}
 t'(v)v+b'(v)=0 ~~\Rightarrow ~~ \beta\frac{2}{v^3}v - \gamma \frac{1}{v^2}=0 ~~ \Rightarrow ~~ 2\beta = \gamma .
\end{equation}

Using the fact that $2\beta =\gamma$, we can now reduce the number of constants to three by replacing $\gamma$.
Then, $\rho$, the approximation factor, will, by definition, be equal to:
\begin{equation}\label{rho}
 \rho = \frac{t(v)v+b(v)}{\frac{1}{3}(v+1)} ~~~ \Rightarrow ~~~ \rho=\frac{3(\alpha v^2-\delta v+\beta)}{v^2+v} .
\end{equation}
Given $t(v)$ and $b(v)$, the fraction of the top item that remains for the {\sc Top} bidder is $1-t(v)$
and the fraction of the bottom item that remains for the {\sc Bottom} bidder is $1-b(v)$.
We need to make sure that these remaining fractions suffice to provide
both the {\sc Top} and the {\sc Bottom} bidders with a $\rho$-approximation
of their PF utilities.
These restrictions translate to:
\begin{align}
 \rho \frac{v+1}{3v} \le 1-t(v)~~~ & \Rightarrow ~~~ t(v)\leq \frac{1}{2}-\frac{b(v)}{2v}  ~~~ & \Rightarrow ~~~ (2\alpha-1)v-\delta \leq 0 \label{TopRestriction}\\
\rho \frac{v+1}{3} \le 1 -b(v)~~~ & \Rightarrow ~~~b(v)\leq \frac{1}{2}-\frac{t(v)v}{2}  ~~~ & \Rightarrow ~~~ \alpha v^2 - (2\delta +1)v+3 \beta \leq 0 .\label{BottomRestriction}
\end{align}


We choose $\alpha = \frac 12 + \frac 14 \delta$, which makes Restriction~(\ref{TopRestriction}) tight.
We choose $\beta = \delta$ which is the smallest value of $\beta$ for which $b(v) \ge 0$ always.
Substituting in Restriction~(\ref{BottomRestriction}) yields:
\begin{equation}\label{DeltaConstraint}
\left(\frac 12 v^2 - v\right) + \delta\left(\frac 14 v^2 - 2v +3\right) \le 0.
\end{equation}
Note that this constraint is tight at $v=2$.
The derivative of the left hand side is $v-1 + \delta (v/2 -2)$, which is negative at $v=1$.
Thus, so long as the constraint is satisfied at $v=1$, it is satisfied for all $v\in [1,2]$.
For this, it suffices that $\frac 54 \delta \le \frac 12$, i.e.\ that $\delta \le \frac 25$.

We choose $\delta = \frac 25$, and then $\beta =\frac 25$ and $\alpha = \frac 35$.
Substituting in Equation~(\ref{rho}) gives:
\begin{equation}\label{rho-with-values}
 \rho =\frac{3\left(\frac 35 v^2-\frac 25 v+\frac 25\right)}{v^2+v} .
\end{equation}
It is a simple matter to check that the derivative of $\rho$ is zero at $(2 + \sqrt{14})/5$
when $\rho \approx 0.89$.

Next, we argue that this is the best choice of parameters.
Consider the choice of parameters $\alpha = \frac 12 + \frac 14 \delta + \alpha'$, where
$\alpha' \le 0$ (Restriction~(\ref{TopRestriction}) at $v=2$ forces $\alpha'\le 0$),
$\beta = \delta + \beta'$ (again, $b(2) \ge 0$ and $\gamma = 2\beta$ forces $\beta' \ge 0$),
and $\delta = \frac 25 + \delta'$ (again, $\delta' \le 0$).
Substituting into Restriction~(\ref{BottomRestriction}) with $v=1$ yields:
\begin{equation}
\label{eqn:other-param}
\alpha' + 3 \beta' - 2 \delta'  \le 0.
\end{equation}
Thus any change to the term $\alpha^2 -\delta v + \beta$ in $\rho$ is
negative: any increase due to $\beta'$ must be offset by a decrease of at least
$3\beta'v^2$ due to $\alpha'$ and any increase due to $-\delta' v$ is similarly
offset by a decrease of $\frac94 \delta' v^2 + \delta'$ due to $\alpha'$, $\alpha$ and $\beta$.

The above scheme handles the case that $R_b$ is the middle bidder.
Otherwise, w.l.o.g., suppose that the bottom bidder is $R_b$
(if not, just switch the roles of the bottom and top items).
Then for $v\le 2$, the bottom bidder is assigned the bottom item, and the other two
bidders each receive half the top item.
For $2 \le v \le \sqrt{12}$,
$t(v) = \frac 14 - \frac{1}{v^2}$ and $b(v) = \frac{2}{v}$,
and the other two bidders each receive $\frac 14 + \frac{1}{v^2}$
of the top item.
For $ v > \sqrt{12}$, all three bidders receive one third of the top item.
The worst approximation factor occurs at $v= \sqrt{12}$ and it has value $(12 -\sqrt{12})/11 \approx 0.77599$.

\end{proof}

\end{document}